\documentclass[sigconf,nonacm]{acmart}
\AtBeginDocument{%
  }

\setcopyright{acmlicensed}
\copyrightyear{2018}
\acmYear{2018}
\acmDOI{XXXXXXX.XXXXXXX}
\acmISBN{978-1-4503-XXXX-X/2018/06}

\usepackage{natbib} 
\usepackage{multirow}           
\usepackage{amsmath, amsfonts}           
\usepackage{graphicx}           
\usepackage{duckuments}         
\usepackage{subcaption}
\usepackage{bm}
\usepackage{xspace}
\usepackage{xurl}
\usepackage{algorithm, algorithmicx, algpseudocode}
\newtheorem{theorem}{Theorem}
\newtheorem{proposition}{Proposition}

\newcommand{\ie}{\emph{i.e.},\xspace}
\newcommand{\eg}{\emph{e.g.},\xspace}

\newcommand{\mmnm}{{CM$^3$}}

\usepackage{enumitem}

\begin{document}

\title{\mmnm{}: Calibrating Multimodal Recommendation}

\author{Xin Zhou}
\affiliation{%
  \institution{Nanyang Technological University}
  \city{Singapore}
  \country{Singapore}}
\email{xin.zhou@ntu.edu.sg}

\author{Yongjie Wang}
\affiliation{%
  \institution{Nanyang Technological University}
  \city{Singapore}
  \country{Singapore}}
\email{yongjie.wang@ntu.edu.sg}

\author{Zhiqi Shen}
\affiliation{%
  \institution{Nanyang Technological University}
  \city{Singapore}
  \country{Singapore}}
\email{ZQShen@ntu.edu.sg}

\begin{abstract}
Alignment and uniformity are fundamental principles within the domain of contrastive learning.
In recommender systems, prior work has established that optimizing the Bayesian Personalized Ranking (BPR) loss contributes to the objectives of alignment and uniformity.
Specifically, alignment aims to draw together the representations of interacting users and items, while uniformity mandates a uniform distribution of user and item embeddings across a unit hypersphere.
This study revisits the alignment and uniformity properties within the context of multimodal recommender systems, revealing a proclivity among extant models to prioritize uniformity to the detriment of alignment. Our hypothesis challenges the conventional assumption of equitable item treatment through a uniformity loss, proposing a more nuanced approach wherein items with similar multimodal attributes converge toward proximal representations within the hyperspheric manifold.
Specifically, we leverage the inherent similarity between items' multimodal data to calibrate their uniformity distribution, thereby inducing a more pronounced repulsive force between dissimilar entities within the embedding space.
A theoretical analysis elucidates the relationship between this calibrated uniformity loss and the conventional uniformity function. Moreover, to enhance the fusion of multimodal features, we introduce a Spherical Bézier method designed to integrate an arbitrary number of modalities while ensuring that the resulting fused features are constrained to the same hyperspherical manifold. Empirical evaluations conducted on five real-world datasets substantiate the superiority of our approach over competing baselines. We also shown that the proposed methods can achieve up to a 5.4\% increase in NDCG@20 performance via the integration of MLLM-extracted features. Source code is available at: \url{https://github.com/enoche/CM3}.
\end{abstract}




\maketitle

\section{Introduction}
The advent of multimodal learning has intensified attention on multimodal recommender systems, which leverage heterogeneous data modalities (\eg visual and textual information) associated with items to achieve effective recommendation~\cite{baltruvsaitis2018multimodal, deldjoo2020recommender,zhou2024disentangled}. Within this burgeoning field, contrastive learning has emerged as a promising paradigm for enhancing the learning of user and item representations from multimodal data. 
In fact, the contrastive learning framework is predicated upon two fundamental principles: alignment and uniformity~\cite{wang2020understanding}. In the context of multimodal recommendation, alignment ensures consistency between representations derived from distinct modalities or positive user-item pairs, while uniformity promotes an equitable distribution of user and item representations across a unit hypersphere. 
Prior research~\cite{wang2022towards}, which relies exclusively on user-item interactions, has established that directly optimizing alignment and uniformity can significantly enhance recommendation performance. However, these principles remain under-explored in multimodal recommendation, where the integration of diverse feature modalities necessitates delicate consideration.

\begin{figure}[bpt]
    \centering
    \advance\leftskip-1cm
    \includegraphics[width=0.8\linewidth]{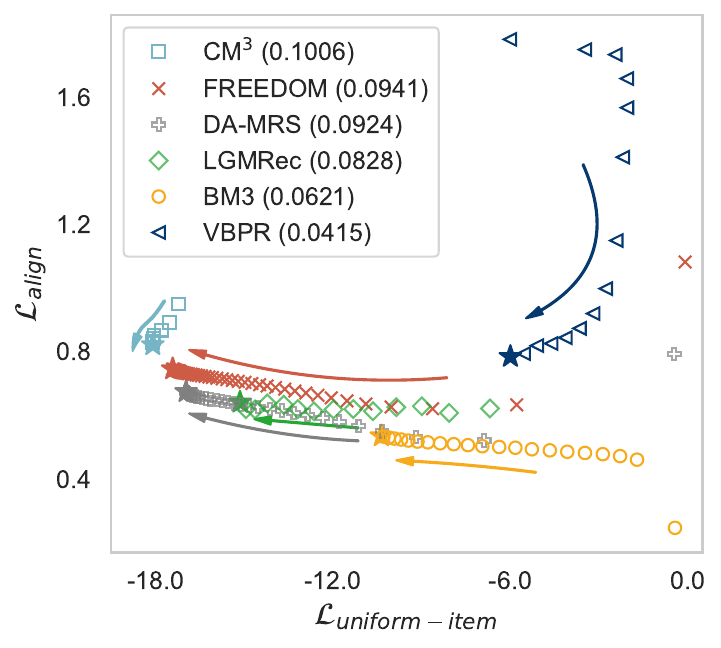}
    \caption{Training dynamics of alignment loss ($l_{\mathrm{align}}$) and item uniformity loss ($l_{\mathrm{uniform-item}}$) for various multimodal models. Optimal validation performance, indicated by stars, is accompanied by corresponding changes in loss values (denoted by colored arrows). Performance is quantified using Recall@20, shown in brackets after each model name.}
    \vspace{-10pt}
    \label{fig:align_uniform}
\end{figure}

\textit{Firstly}, we demonstrate the \textbf{contrary} in optimizing alignment and uniformity in current multimodal recommender models. 
According to Theorem 1 of DirectAU~\cite{wang2022towards}, perfectly aligned and uniform encoders, if they exist, are the global minimizers of the Bayesian Personalized Ranking (BPR) loss~\cite{rendle2012bpr}. This implies that the BPR loss inherently promotes lower alignment between positive user-item pairs and uniformity between user-user and item-item pairs.
However, our empirical analysis reveals that this theoretical optimum is not achieved in multimodal recommendation. We illustrate this by plotting the training evolution of alignment and uniformity metrics for five representative multimodal models with Clothing dataset in Fig.~\ref{fig:align_uniform}. Among these, BM3~\cite{zhou2023bootstrap} utilizes contrastive learning techniques to derive user and item representations, while VBPR~\cite{he2016vbpr} and FREEDOM~\cite{zhou2023tale} employ the BPR loss for model optimization. It is noteworthy that LGMRec~\cite{guo2024lgmrec} and DA-MRS~\cite{xv2024improving} incorporates both contrastive learning and BPR loss for model optimization.
As depicted in Fig.~\ref{fig:align_uniform}, all multimodal models exhibited a distinct bias towards optimizing uniformity, thereby compromising alignment, during the latter stages of training. 
This unexpected finding reveals a divergence from the typical optimization behavior seen in general recommender systems, as documented in~\cite{wang2022towards}.

\textit{Secondly}, we delve into the underlying mechanisms that precipitate the observed behavior. Given two pairs of interactions for $u$ as $(u, i)$ and $(u, j)$, based on the research of~\cite{wang2020understanding}, $l_{\mathrm{align}}$ minimizes both $\mathbb{E}_{(u, i) \sim p_{\text {pos }}} \left[\|f(u)-f(i)\|_2^2\right]$ and $\mathbb{E}_{(u, j) \sim p_{\text {pos }}}\left[\|f(u)-f(j)\|_2^2\right]$, while $l_{\mathrm{uniform}}$ minimizes $\mathbb{E}_{(i,j) \sim p_{\text {item }}}\left[e^{-t\|f(i)-f(j)\|_2^2}\right]$.
If $u$ is perfectly aligned with both items, $\mathbb{E}_{(i, j) \sim p_{\text {item }}}\left[\|f(i)-f(j)\|_2^2\right]$ tends to be minimized. However, this conflicts with the objective of $l_{\mathrm{uniform}}$, which aims to maximize $\|f(i)-f(j)\|_2^2$. 
Consequently, models face challenges in balancing the optimization of these competing objectives.
Furthermore, the incorporation of multimodal information further deteriorates item uniformity optimization, as items with similar multimodal features cluster more tightly in the embedding space than items with randomly generated multimodal data,
as can be evidenced by Table~\ref{tab:mm_uni_random} in the Appendix. 

To address this issue, we propose a \underline{C}alibrated \underline{M}ulti\underline{M}odal \underline{M}odel (\mmnm{}) that enhances recommendation efficacy by modulating item uniformity via the utilization of multimodal information. Specifically, we initially compute a similarity score based on the multimodal features of items. This score is subsequently integrated into the uniformity loss, aiming to repel dissimilar items while maintaining proximity between similar items. We further provide a theoretical analysis demonstrating the pivotal role of the similarity score in determining the behavior of the calibrated uniformity loss with respect to items.
To quantify similarity by leveraging the intrinsic information of each modality, we propose a Spherical Bézier fusion method that integrates multimodal data into a unified vector. The item-item similarity score is then derived from this composite vector. This approach ensures that the resulting vectors retain hyperspherical properties, as each constituent modality vector already lies on the hypersphere.
Our key contributions are as follows:
\begin{itemize} [leftmargin=*]
\item We elucidate the inherent dilemma faced by conventional multimodal recommendation models in simultaneously optimizing the alignment of positive interactions and maintaining uniformity between user-user and item-item relationships.
\item We introduce a novel calibrated recommendation model, \mmnm{}, which refines inter-item relations within the uniformity loss function by utilizing multimodal features. In \mmnm{}, we design a spherical Bézier fusion method to blend data from all modalities, preserving semantics by integrating multimodal features along the shortest path on a spherical surface.
\item We conduct comprehensive empirical evaluations on real-world datasets, demonstrating that \mmnm{} significantly outperforms state-of-the-art multimodal recommender systems. To gain a nuanced understanding of \mmnm{}'s efficacy, we also perform extensive ablation studies under various evaluation configurations.
\end{itemize} 

\section{Related Work}
\subsection{Multimodal Recommendation}
Multimodal recommendation leverages multimodal information (\eg images and textual descriptions) of items to enhance the recommendation performance within the collaborative filtering paradigm~\cite{deldjoo2020recommender,zhou2023comprehensive,zhang2024multiH,liuqijiong2024multimodal,zhu2024multimodal,liu2024multimodal}. Early studies~\cite{he2016vbpr, liu2017deepstyle, chen2019personalized, liu2019user} adopted deep learning techniques to extract visual and/or textual features of items, along with the original item embeddings, to model user-item interactions within the BPR framework~\cite{rendle2012bpr}. With the help of multimodal information, these methods could better capture user preferences. Graph Neural Networks (GNNs), which capture high-order structures in user-item interactions, have successfully enhanced user and item representations by aggregating multi-hop neighborhood information, as demonstrated in later studies~\cite{wei2019mmgcn, wei2020graph, wang2021dualgnn,zhou2023enhancing,LayerGCN}. LATTICE~\cite{zhang2021mining} highlights that incorporating item-item relationships can enhance item representations. To achieve this, it first learns item-item graphs for each modality and then fuses these graphs into a final item-item graph. FREEDOM~\cite{zhou2023tale} argues that item-item graph learning is trivial and introduces computational overhead in LATTICE~\cite{zhang2021mining}. To address this, it freezes the item-item graphs and further denoises user-item graphs for more efficient and effective recommendations. LGMRec~\cite{guo2024lgmrec} jointly learns local and global representations of users and items to model user-item interactions at multiple granularities. 
PGL~\cite{yu2025mind} effectively extracts and leverages principal local structural features from user-item interaction graphs to enhance graph learning, delivering superior recommendation performance.
SMORE~\cite{ong2025spectrum} fuses multi-modal features in the spectral domain, suppresses modality-specific noise with an adaptive filter.
Another line of research~\cite{tao2022self, zhou2023bootstrap, wei2023multi, yi2022multi, wang2024multi} adopts a self-supervised learning framework with contrastive learning by augmenting multi-view data to address the data scarcity problem. 

Our study distinguishes itself from existing methods in the field of multimodal recommendation by directly exploiting alignment and uniformity losses. In contrast, previous works typically employ contrastive learning loss as a complementary objective, often combining it with other loss functions. This fundamental difference in methodology allows our model to more explicitly optimize for both alignment and uniformity in the multimodal representation space, potentially leading to more robust and effective recommendations.
We anticipate that our study will stimulate further research in related domains~\cite{zhou2025crowd}, including sequential recommendation~\cite{WhitenRec,DWSRec} and sustainable recommendation systems~\cite{zhou2024advancing}.

\subsection{Contrastive Learning} 
Contrastive learning (CL) has demonstrated remarkable success across various domains~\cite{chen2020simple,gao2021simcse,li2019graph,radford2021learning,GDCL,MP4SR,zhang2024multiC,qu2025enhancing,qiu2024paircfr}. The objective of CL is to map semantically similar data to closely aligned embeddings while separating semantically dissimilar data into distinct regions of the embedding space~\cite{chopra2005learning, schroff2015facenet, oord2018representation}. A common approach for stabilizing CL training is to normalize latent representations onto the unit hypersphere. Empirical studies have shown that normalized representations outperform unnormalized counterparts, such as those in Euclidean space~\cite{schroff2015facenet, wang2017normface}. As stated by~\cite{wang2020understanding}, minimizing contrastive loss on normalized space is equivalent to minimizing two objectives: 1) alignment, where samples from positive pairs should have similar features; and 2) uniformity, where feature vectors of all data points should be roughly uniformly distributed on the unit hypersphere. 

Following~\cite{wang2020understanding}, recent research~\cite{wang2022towards, gao2021simcse} directly optimizes the alignment and uniformity terms to avoid the need for hard example sampling. 
However, we observe that this finding does not hold in multimodal recommendation. The tie is broken by simultaneously optimizing alignment between interacted users and items, and uniformity within item-item and user-user pairs.  In this work, we propose a novel calibrated uniformity loss for items, specifically designed for multimodal recommendation scenarios, to address the optimization conflicts between alignment and uniformity terms.

\begin{figure*}[hbtp]
    \centering
    \includegraphics[width=\linewidth,trim={90 10 10 10},clip]{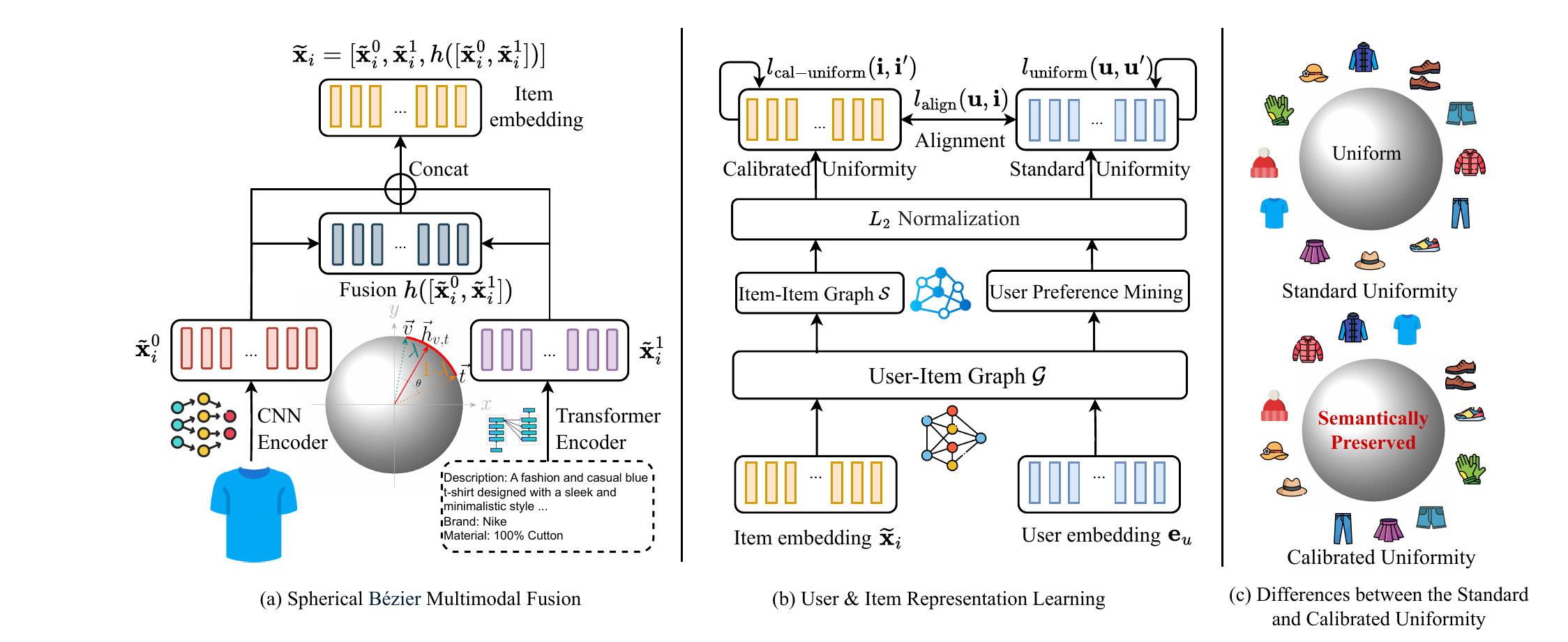}
    \caption{Overview of our proposed \mmnm. (a) We encode the product images and textual descriptions using pre-trained models, then mix multimodal features with a Spherical Bézier Fusion. (b) Initial item and user embeddings are enhanced through user-item and item-item graphs using GNNs. Alignment, as well as standard and calibrated uniformity losses, are used to optimize the distributions of user and item representations on a unit hypersphere. (c) A toy example demonstrates the differences between standard uniformity and our calibrated uniformity losses over items.}
    \label{fig:framework}
\end{figure*}

\section{Calibrated Multimodal Recommendation}
\label{sec:model}
\subsection{Overview of \mmnm{}}
The crux of multimodal models lies in their capacity to learn informative user and item representations for recommendation, leveraging rich multimodal features. To this end, \mmnm{} implements a bifurcated strategy to derive item representations: i) the augmentation of multimodal features through a Spherical Bézier Multimodal Fusion technique, which facilitates the integration and transformation of diverse data modalities in a hyperspherical manifold space; and ii) the refinement of low-dimensional item embeddings via the application of a meticulously calibrated uniformity loss function. This novel uniformity loss enables the model to distill the multifaceted nature of multimodal information into a refined and discriminative representational framework, thereby enhancing the model's capacity to capture nuanced item characteristics and inter-item relationships.

We elucidate the constituent components of \mmnm{} in the subsequent subsections, accentuating the innovative design elements while concisely referencing the foundational mechanisms upon which \mmnm{} is constructed, such as Graph Convolutional Networks (GCN). Fig.~\ref{fig:framework} presents an overview of \mmnm{}.
\subsection{Notations}
Consider a dataset $\mathcal{D}$ defined by the tuple $\mathcal{D} = \{\mathcal{G}, \mathbf{X}^0, \cdots, \mathbf{X}^{|\mathcal{M}|-1}\}$, where $\mathcal{G} = (\mathcal{V}, \mathcal{E})$ denotes the interaction bipartite graph.
The set $\mathcal{M}$ encompasses all available modalities pertinent to the items under consideration. Within this framework, $\mathcal{E}$ and $\mathcal{V}$ denote the edge set and node set of the graph, respectively, encapsulating the interactions between users and items.
More formally, an edge $\mathcal{E}_{ui}=1$ within $\mathcal{G}$ signifies the existence of an interaction between a user $u$ and an item $i$. The node set $\mathcal{V}$ is defined as the union of user and item sets, such that $\mathcal{V}=\mathcal{U} \cup \mathcal{I}$, where $u \in \mathcal{U}$ represents a user and $i \in \mathcal{I}$ denotes an item.
For each modality $m \in \mathcal{M}$, we define a feature matrix $\mathbf{X}^m \in \mathbb{R}^{|\mathcal{I}| \times d_m}$, where $|\mathcal{I}|$ represents the cardinality of the item set and $d_m$ signifies the original dimensionality of the feature space for modality $m$.

Given the dataset $\mathcal{D}$, the objective of a multimodal recommender system is to generate a ranked list of items for each user $u$, predicated on a preference score function. This function, denoted as $\hat{y}_{ui}$, quantifies the predicted affinity between user $u$ and item $i$, and is formally defined as:
$\hat{y}_{ui} = f_{\Theta}(u, i, \mathbf{x}^0_i, \cdots, \mathbf{x}^{|\mathcal{M}|-1}_i)$.
The model $f_{\Theta}(\cdot)$ is parameterized by $\Theta$, a set of trainable parameters.

\subsection{Spherical Bézier Multimodal Fusion}
\subsubsection{Multimodal Feature Projection}
The multimodal features extracted from pretrained models are often tangentially related to the downstream task and typically characterized by high dimensionality. To address these challenges, we employ Deep Neural Networks (DNNs) to project each individual modality feature into its corresponding low-dimensional space. This dimensionality reduction not only mitigates computational complexity but also enhances the relevance of the features to the task at hand.
Specifically, given a unimodal feature matrix of items, denoted as $\mathbf{X}^m \in \mathbb{R}^{|\mathcal{I}| \times d_m}$, we derive the latent unimodal representation through the following equation:
\begin{equation}
    \widetilde{\mathbf{X}}^m = \sigma({\mathbf{X}^m \mathbf{W}_1^m  + \mathbf{b}_1^m}) \mathbf{W}^m_2,
    \label{eq:mm}
\end{equation}
where $\sigma(\cdot)$ denotes an activation function, such as the `$\mathrm{Leaky\_relu}$' function. $\mathbf{W}_1^m \in \mathbb{R}^{d_m \times d_1}$, $\mathbf{W}_2^m \in \mathbb{R}^{d_1 \times d}$, and $\mathbf{b}_1^m \in \mathbb{R}^{d_1}$ represent the trainable weight matrices and bias vector, respectively. Here, $d_1$ and $d$ indicate the vector dimensions.

\subsubsection{Infinite Multimodal Fusion}
Given the unimodal representations derived from Equation~\eqref{eq:mm} using distinct pre-trained encoders, a multimodality gap may arise. To address this, we propose an advanced interpolation method based on Mixup to effectively fuse the representations. Mixup~\cite{zhang2018mixup, verma2019manifold, oh2024geodesic} is a technique that linearly interpolates pairs of data points, creating synthetic samples to enrich the training set. Empirical studies have consistently demonstrated its effectiveness in improving the generalization and robustness of neural networks.
While traditional Mixup typically leverages both feature vectors and labels from two samples for interpolation. In this work, we extend this approach to enable infinite multimodal fusion. First, in the absence of labels, we interpolate multimodal features corresponding to the same item (\textit{item as label}) for fusion. Second, we employ De Casteljau's algorithm to iteratively combine an infinite number of multimodal features. This method ensures that the interpolated vector traverses a Bézier curve defined by the multimodal vectors while remaining constrained within the hyperspherical manifold.
Given a set of unimodal features $[\tilde{\mathbf{x}}^m_i]$, where $m \in \mathcal{M}$, for an item $i$, the mixed feature can be computed as:
\begin{equation}
		h([\tilde{\mathbf{x}}^m_i]) = \underbrace{f(\tilde{\mathbf{x}}^{|\mathcal{M}|-1}_i, f(\cdots, f(\tilde{\mathbf{x}}^2_i, f(\tilde{\mathbf{x}}^1_i, \tilde{\mathbf{x}}^0_i))\cdots ))}_{|\mathcal{M}|-1},
		\label{eq:xmixup}
\end{equation}
where $f(\vec{a}, \vec{b})$  denotes the spherical interpolation function that is defined as:
\begin{equation}
		f(\vec{a}, \vec{b}) = \frac{\mathrm{sin} (\lambda \theta)}{\mathrm{sin}(\theta)} \vec{a} + \frac{\mathrm{sin} ((1-\lambda) \theta)}{\mathrm{sin}(\theta)} \vec{b}.
		\label{eq:geomix}
\end{equation}

In this equation, $\theta = \cos^{-1}(\vec{a}, \vec{b})$ represents the angle between vectors $\vec{a}$ and $\vec{b}$, and $\lambda$ is sampled from a Beta distribution with hyperparameter $\alpha$, such that $\lambda \sim \text{Beta}(\alpha, \alpha)$.
\begin{proposition}
Given that all vectors in $[\tilde{\mathbf{x}}_i^m]$ lie on the hypersphere, the mixed feature defined by Equation~\eqref{eq:xmixup} also lies on the hypersphere.
 \label{th:xmix_hyper}
\end{proposition}
\begin{proof}
The proof of this proposition is straightforward and is provided in the Appendix~\ref{append:a}.
\end{proof}

\subsection{Enhancing User and Item Representations via Graph Learning}
To adeptly capture higher-order interactions between users and items, as well as the intricate semantic relationships among items, we employ the widely acknowledged graph learning paradigm~\cite{zhang2021mining, zhou2023tale}. This methodology facilitates the derivation of user and item representations from both user-item and item-item graphs.
\subsubsection{Graph Learning on User-Item Graph}
We first concatenate the $|\mathcal{M}| + 1$ latent features into a single vector to signify item representation:
\begin{equation}
    \widetilde{\mathbf{X}} = \mathrm{Concat}(\widetilde{\mathbf{X}}^{0}, \cdots, \widetilde{\mathbf{X}}^{|\mathcal{M}|-1}, h([\widetilde{\mathbf{X}}^m])),
    \label{eq:itemrep}
\end{equation}
where $h([\widetilde{\mathbf{X}}^m])$ is the mixed features via Equation~\eqref{eq:xmixup} at matrix view.
The dimension of $\widetilde{\mathbf{X}}$ is $\mathbb{R}^{|\mathcal{I}| \times \ell}$, where $\ell =|\mathcal{M}|d+d$.

To accommodate user preference and attend to both the modality-specific latent feature and the multi-modality shared feature, we formulate a user ID embedding matrix, represented as $\mathbf{E} \in \mathbb{R}^{|\mathcal{U}| \times \ell}$.
For the propagation of information within the convolutional network, we employ LightGCN~\cite{he2020lightgcn}. Specifically, the representations of user $u$ and item $i$ at the $(l+1)$-th graph convolution layer of $\mathcal{G}$ are derived as follows:
\begin{equation}
	\begin{split}
        \mathbf{e}_{u}^{(l+1)} &= \sum_{i\in \mathcal{N}_{u} } \frac{1}{\sqrt{\left | \mathcal{N}_{u}  \right | }\sqrt{\left | \mathcal{N}_{i}  \right | } } \widetilde{\mathbf{x}}_{i}^{(l)};\\ 
        \widetilde{\mathbf{x}}_{i}^{(l+1)} &= \sum_{u\in \mathcal{N}_{i} } \frac{1}{\sqrt{\left | \mathcal{N}_{u}  \right | }\sqrt{\left | \mathcal{N}_{i}  \right | } } \mathbf{e}_{u}^{(l)},
        \end{split}
\end{equation}
where $\mathcal{N}_u$ and $\mathcal{N}_i$ denote the set of first hop neighbors of $u$ and $i$ in $\mathcal{G}$, respectively.
Employing $L_{ui}$ layers of convolutional operations, we extract all representations from the hidden layers to formulate the final representations for both users and items:
\begin{equation}
	\begin{split}
		\mathbf{E} &= \mathrm{R{\scriptsize EADOUT}}(\mathbf{E}^0, \mathbf{E}^1, \cdots ,\mathbf{E}^{L_{ui}}); \\
		\widetilde{\mathbf{X}} &= \mathrm{R{\scriptsize EADOUT}}(\widetilde{\mathbf{X}}^0, \widetilde{\mathbf{X}}^1, \cdots, \widetilde{\mathbf{X}}^{L_{ui}}),
		\label{eq:lgn_layer_update}
	\end{split}
\end{equation}
where the $\mathrm{R{\scriptsize EADOUT}}$ function can be any differentiable function. We use the sum function to derive the representations.

\subsubsection{User Preference Mining}
To distinguish user preferences among multimodal features, we partition user embeddings into $|\mathcal{M}|+1$ segments, each corresponding to the modality features as defined in Equation~\eqref{eq:itemrep}. A learnable weight matrix $\mathbf{W}_3 \in \mathbb{R}^{|\mathcal{U}|\times (|\mathcal{M}|+1) \times 1}$ is initialized and employed to compute the final user representation. Following the reshaping of both $\mathbf{E}$ to the dimensions of $\mathbb{R}^{|\mathcal{U}|\times (|\mathcal{M}|+1) \times d}$, we calculate:
\begin{equation}
    \widehat{\mathbf{E}} = \mathbf{W}_3 \mathbf{E}.
    \label{eq:userrep_f}
\end{equation}
Subsequently, having obtained the differentiated user preferences, we reshape the representation back to its original dimensions as:
$\widehat{\mathbf{E}}=\widehat{\mathbf{E}}.\mathrm{view}(|\mathcal{U}|, \ell)$.

\subsubsection{Graph Learning on Item-Item Graph}
To further elucidate the high-order relationships between items, we adhere to the methodology outlined in existing work~\cite{zhou2023tale} to construct an item-item graph $\mathbf{S}$ based on multimodal features. Subsequently, we perform graph convolutions on the items using the item-item graph $\mathbf{S}$ to derive the final item representations.
The detailed procedures involved in this process are not elaborated upon here, as they are analogous to those used in the user-item graph.
With the last layer's representation $\widetilde{\mathbf{X}}^{L_{ii}}$ ($L_{ii}$ is the number of layers), we establish a residual connection with the initial item representation ($\widetilde{\mathbf{X}}^{0}$) to obtain the item final representation:
\begin{equation}
	\widehat{\mathbf{X}} = \widetilde{\mathbf{X}}^{L_{ii}} + \widetilde{\mathbf{X}}^{0}.
	\label{eq:ii_emb}
\end{equation}

\subsection{Alignment and Calibrated Uniformity}
Consider a positive pair $(u, i)$ of user and item with corresponding representations $\mathbf{u}$ and $\mathbf{i}$, respectively. The alignment and uniformity losses are defined as follows:
\begin{equation}
    \begin{split}
        l_{\rm align}(u,i) =& \mathop{\mathbb{E}}_{(\mathbf{u}, \mathbf{i})\sim p_{\rm pos}}||\mathbf{u} - \mathbf{i}||^2 ;\\
        l_{\rm uniform}(i, i') =
        &~\log\mathop{\mathbb{E}}_{\mathbf{i}, \mathbf{i}'\sim p_{\rm item}}e^{-t ||\mathbf{i} - \mathbf{i}'||^2},
    \end{split}
    \label{eq:audef}
\end{equation}
where $t > 0$ is a temperature parameter, $p_{\text{pos}}$, $p_{\text{user}}$, and $p_{\text{item}}$ denote the distributions of positive user-item pairs, users, and items, respectively. $\mathbf{u}'$ and $\mathbf{i}'$ represent the embeddings of user $u'$ and item $i'$.
The alignment loss serves to bring positive pairs $(u, i)$ closer in the embedding space, while the uniformity loss repels users from other users and items from other items, promoting a uniform distribution of representations.

We propose that the relationships between items should be differentiated. Consequently, we modify the uniformity loss for items as follows:
\begin{equation}
        l_{\rm cal-uniform}(i,i') = \log\mathop{\mathbb{E}}_{\mathbf{i}, \mathbf{i}' \sim p_{\rm item}}e^{-t \left(||\mathbf{i} - \mathbf{i}'||^2 - 2 + 2{s}(\bar{\mathbf{i}}, \bar{\mathbf{i}}') \right)},
    \label{eq:softu}
\end{equation}
where $s(\cdot)$ is a function that computes a clamped similarity score between two vectors, which can be pre-calculated before loss computation. $\bar{\mathbf{i}}$ represents any level of representation for item $i$. In this context, we utilize the mixed features of $i$, defined as $\bar{\mathbf{i}} = h([\tilde{\mathbf{x}}^m_i])$.

\begin{theorem}[Calibrated Uniformity Amplification]
Let $\mathcal{I}$ be the set of all items, and let $\varphi=s(\bar{\mathbf{i}}, \bar{\mathbf{i}}')$ denote the similarity between a specific pair of items $i, i' \in \mathcal{I}$. Consider the calibrated uniformity loss function $l_{\mathrm{cal-uniform}}$ defined above, the following statement holds:

The calibrated uniformity loss $l_{\mathrm{cal-uniform}}$ amplifies the repulsion between items $i$ and $i'$ by a factor of $e^{2t(1-\varphi)}$ relative to the standard uniformity loss.

\label{th:softuni1}
\end{theorem}

\begin{proof}
Note that for $\mathbf{i}, \mathbf{i}' \in \mathcal{S}^d$, where $\mathcal{S}^d$ is a unit hypersphere, we have: $||\mathbf{i} - \mathbf{i}'||^2 = 2 - 2 \cdot \mathbf{i}^\top \mathbf{i}'$.

Relation between $l_{\mathrm{cal-uniform}}$ and $l_{\mathrm{uniform}}$:
\begin{align}
    \frac{e^{-t\left(||\mathbf{i} - \mathbf{i}'||^2 - 2 + 2 \varphi \right)}}{e^{-t\left(||\mathbf{i} - \mathbf{i}'||^2 \right)}} &= \frac{e^{-t\left(2 - 2 \cdot \mathbf{i}^\top \mathbf{i}' - 2 + 2 \varphi \right)}}{e^{-t(2 - 2 \cdot \mathbf{i}^\top \mathbf{i}' )}} = e^{2t(1-\varphi)}
\end{align}

Given the clamped similarity score between items is bounded within the interval $[0, 1]$, the calibrated uniformity loss $l_{\mathrm{cal-uniform}}$ degenerates to the standard uniformity loss $l_{\mathrm{uniform}}$ iff $\varphi=1$. Conversely, for $\varphi \neq 1$, $l_{\mathrm{cal-uniform}}$ imposes a more stringent repulsion ($\because$ $e^{1-\varphi} > 1$) between dissimilar items compared to the standard uniformity loss. Consequently, this mechanism promotes items that are similar to themselves to be positioned closer together on the hypersphere.
\end{proof}

\subsection{Model Optimization and Recommendation}
For model optimization, we adopt the alignment for positive pairs and the uniformity loss for users with representation of $\widehat{\mathbf{E}}$, but with the calibrated uniformity loss on items based on representation of $\widehat{\mathbf{X}}$. The final loss:
\begin{equation}
    \mathcal{L} = l_{\rm align}(u,i) + \gamma \left(l_{\rm uniform}(u, u') + l_{\rm cal-uniform}(i, i')\right).
    \label{eq:auopt}
\end{equation}

To generate item recommendations for a user, we calculate the score for a possible interaction between $u$ and $i$ as:
\begin{equation}
	\hat{y}_{ui} = \widehat{\mathbf{e}}_u^\top \widehat{\mathbf{x}}_i.
\end{equation}
A high score suggests that the user prefers the item. Based on these scores, we select the top-$k$ items as recommendations for user $u$.

\subsection{Computational Complexity Analysis}
The computational complexity associated with the alignment and uniformity computations is equivalent for both DirectAU and \mmnm{}, with the exception of similarity score calculations ($\mathcal{O} \bigl( \sum_m^\mathcal{M} |\mathcal{I}| (d_m d_1 + d_1 d)+d^2) \bigl)$) and graph learning of item-item Graph ($\mathcal{O} (L_{ii} |\mathcal{I}|)$). 
Given that the additional computational cost does not substantially increase runtime relative to the baseline DirectAU, we can infer that \mmnm{}'s computational complexity is of the same order as DirectAU.

\section{Experiment Settings}
\subsection{Datasets}
\label{sec:datasets}
Following existing research \cite{he2016vbpr, zhang2021mining, zhou2023bootstrap, zhou2023tale}, we conduct experiments on the Amazon review dataset, which contains both product descriptions and multi-view images. 
The multimodal information inherent in these datasets provides an ideal context for the rigorous evaluation of multimodal recommendation algorithms. Our experimental design incorporates four distinct category-specific datasets: Baby, Sports, Clothing, and Electronics. To ensure data quality and relevance, we applied a 5-core filtering process to both item and user data, effectively removing entries with insufficient interactions. 
To further investigate the generalization capabilities of our model, we employ the MicroLens dataset~\cite{ni2023content}, which comprises data collected from a short-video platform.
The key statistical characteristics of these refined datasets are summarized in Table~\ref{tab:datasets}, offering a quantitative overview of the data used in our experimental procedures.
For the utilization of multimodal information from Amazon datasets, we adhered to established preprocessing protocols as described in~\cite{zhou2023tale, zhou2025learning}.

\begin{table}[bpt]
	\centering	
	\def\arraystretch{0.9}	
	\caption{Statistics of the experimental datasets.}
	\begin{tabular}{l r r r r}
		\toprule
		Dataset & \# Users & \# Items & \# Interactions & Sparsity \\
		\midrule
		Baby & 19,445 & 7,050 & 160,792 & 99.88\% \\
		Sports & 35,598 & 18,357 & 296,337 & 99.95\%\\
		Clothing & 39,387 & 23,033 & 278,677 & 99.97\%\\
		Electronics & 192,403 & 63,001 & 1,689,188 & 99.99\%\\
            MicroLens & 98,129 & 17,228 & 705,174 & 99.96\%\\
		\bottomrule
	\end{tabular}
	\vspace{-16pt}
	\label{tab:datasets}
\end{table}

\subsection{Baselines}
To demonstrate the efficacy of our proposed method, we conduct a comprehensive comparison against the following widely-adopted baselines in general CF models (\ie \textbf{MF}~\cite{rendle2012bpr}, \textbf{LightGCN}~\cite{he2020lightgcn}, \textbf{SelfCF}~\cite{SelfCF23}, \textbf{DirectAU}~\cite{wang2022towards}) and multimodal recommendation (\ie \textbf{VBPR}~\cite{he2016vbpr}, \textbf{MMGCN}~\cite{wei2019mmgcn}, \textbf{GRCN}~\cite{wei2020graph}, \textbf{LATTICE}~\cite{zhang2021mining}, \textbf{SLMRec}~\cite{tao2022self}, \textbf{BM3}~\cite{zhou2023bootstrap}, \textbf{FREEDOM}~\cite{zhou2023tale}, \textbf{LGMRec}~\cite{guo2024lgmrec}, \textbf{DA-MRS}~\cite{xv2024improving}, \textbf{MIG-GT}~\cite{hu2025modality}).
We briefly summarize their key points as follows:
\textbf{MF}~\cite{rendle2012bpr} utilizes BPR loss to enhance latent representations of users and items within a matrix factorization framework. \textbf{LightGCN}~\cite{he2020lightgcn} incorporates a simplified GCNs to derive item and user representations through neighbor information aggregation and propagation.
\textbf{SelfCF}~\cite{SelfCF23} employs three contrastive view perturbations within a self-supervised learning paradigm to generate latent representations of items and users. The ``embedding dropout'' method from SelfCF is adopted here due to its reported superior performance.
\textbf{DirectAU}~\cite{wang2022towards} establishes a direct link between standard BPR loss and the minimization of alignment and uniformity, proposing a simple yet effective approach to optimize these properties for enhanced recommendation performance. \textbf{VBPR}~\cite{he2016vbpr} extracts visual representations using pre-trained CNNs and concatenates these with item embeddings to model user preferences. \textbf{MMGCN}~\cite{wei2019mmgcn} leverages GCNs on modality-specific interaction graphs to derive user preferences in recommendation tasks. \textbf{GRCN}~\cite{wei2020graph} employs user preference and item content affinity to refine the user interaction graph, aiming to mitigate false-positive interactions and prevent noise propagation along edges. \textbf{LATTICE}~\cite{zhang2021mining} models item-item relationships across feature modalities, fusing them to construct a semantic item-item graph. GCNs are applied to both the fused item-item and user-item graphs for more effective embedding learning. 
\textbf{SLMRec}~\cite{tao2022self} advances multimedia recommendation by using self-supervised learning to capture richer user and item relationships, leading to more accurate recommendation performance.
\textbf{BM3}~\cite{zhou2023bootstrap} augments latent representations of items and users through a dropout strategy, introducing a novel self-supervised learning approach to derive high-quality user and item representations. \textbf{FREEDOM}~\cite{zhou2023tale} addresses the limitations of LATTICE by proposing to freeze the item-item graph and further denoising the user-item graphs, enhancing computational efficiency and representation quality. \textbf{LGMRec}~\cite{guo2024lgmrec} simultaneously learns local and global user interests for effectively recommendation. The local graph captures collaborative and multimodal embeddings, while the global graph represents multiple user group interests, addressing sparsity issues in recommendations.
\textbf{DA-MRS}~\cite{xv2024improving} introduces a denoising and alignment framework designed to mitigate noise within multimodal content and user feedback, while also facilitating their alignment through fine-grained guidance.
\textbf{MIG-GT}~\cite{hu2025modality} aims to integrate information from various data modalities using graph neural networks, enhanced by global transformers to capture broader dependencies and improve recommendation accuracy.

\begin{table*}[bpt]
    \centering
    \def\arraystretch{0.86}	
    \setlength{\tabcolsep}{7.5pt}
    \caption{Performance comparison of different recommendation methods in terms of Recall@20 and NDCG@20. The best results are indicated in \textbf{bold} text, and the second-best results are \underline{underlined}. `*' denotes that the improvements (\textit{Imp.}) are statistically significant compared of the best baseline in a paired $t$-test with $p<0.05$. }
    \begin{tabular}{l| cc | cc | cc | cc | cc}
        \toprule
        Dataset & \multicolumn{2}{c|}{Baby} & \multicolumn{2}{c|}{Sports} & \multicolumn{2}{c|}{Clothing} & \multicolumn{2}{c|}{Electronics} & \multicolumn{2}{c}{Microlens}\\
        \midrule
        Metric & R@20 & N@20 & R@20 & N@20 & R@20 & N@20 & R@20 & N@20 & R@20 & N@20\\
        \midrule
        MF & 0.0575 & 0.0249 & 0.0653 & 0.0298 & 0.0303 & 0.0138 & 0.0367 & 0.0161 & 0.0959 & 0.0408 \\
        LightGCN & 0.0754 & 0.0328 & 0.0864 & 0.0387 & 0.0544 & 0.0243 & 0.0540 & 0.0250 & 0.1075 & 0.0467 \\
        SelfCF & 0.0822 & 0.0357 & 0.0955 & 0.0427 & 0.0616 & 0.0275 & 0.0653 & 0.0306  &  0.1125 &  0.0473 \\
        DirectAU & 0.0804 & 0.0367 & 0.1017 & 0.0464 & 0.0669 & 0.0298 & {0.0666} & {0.0315}  & 0.1186 & \underline{0.0524}\\
        \midrule
        VBPR & 0.0663 & 0.0284 & 0.0856 & 0.0384 & 0.0415 & 0.0192 & 0.0458 & 0.0202  & 0.1026 & 0.0441 \\
        MMGCN & 0.0660 & 0.0282 & 0.0636 & 0.0270 & 0.0361 & 0.0154 & 0.0331 & 0.0141  & 0.0701 & 0.0279 \\
        GRCN  & 0.0824 & 0.0358 & 0.0919 & 0.0413 & 0.0657 & 0.0284 & 0.0529 & 0.0241  & 0.1070 & 0.0460 \\
        LATTICE & 0.0850 & 0.0370 & 0.0953 & 0.0421 & 0.0733 & 0.0330 & OOM & OOM  & 0.1089 & 0.0473\\
        SLMRec & 0.0810 & 0.0357 & 0.1017 & 0.0462 & 0.0810 & 0.0357 & 0.0651 & 0.0303  & 0.1190 &  0.0511 \\
        BM3  & 0.0883 & 0.0383 & 0.0980 & 0.0438 & 0.0621 & 0.0281 & 0.0648 & 0.0302  & 0.0981 & 0.0400 \\
        FREEDOM  & 0.0992 & 0.0424 & {0.1089} & {0.0481} & \underline{0.0941} & {0.0420} & 0.0601 & 0.0273  & 0.1032 & 0.0437 \\
        LGMRec & {0.1002} & {0.0440} & 0.1068 & 0.0480 & 0.0828 & 0.0371 & 0.0625 & 0.0287  & 0.1132 & 0.0489  \\
        DA-MRS & {0.0966} & {0.0426} & {0.1078} & {0.0475} & {0.0924} & {0.0415} & OOM & OOM  & \underline{0.1196} & {0.0520} \\
        MIG-GT & \underline{0.1021} & \underline{0.0452} & \underline{0.1130} & \underline{0.0511} & {0.0934} & \underline{0.0422} & \underline{0.0696} & \underline{0.0320}  & {0.1189} & {0.0523} \\
        \midrule
        \textbf{\mmnm{}} & \textbf{0.1034} & \textbf{0.0470}* & \textbf{0.1222}* & \textbf{0.0567}* & \textbf{0.1006}* & \textbf{0.0463}* & \textbf{0.0760}* & \textbf{0.0359}*  & \textbf{0.1258}* & \textbf{0.0554}* \\
        \textbf{\textit{Imp.}} & {1.27\%} & {3.98\%} & {8.14\%} & {10.96\%} & {6.91\%} & {9.72\%} & {9.20\%} & {12.19\%}  & {5.18\%} & {5.73\%} \\
        \bottomrule
        \multicolumn{11}{l}{- `OOM' denotes an Out-Of-Memory condition encountered on a Tesla V100 GPU with 32 GB of memory. }
    \end{tabular}
    \label{tab:perform}
    \vspace{-10pt}
\end{table*}

\subsection{Evaluation Metrics and Scenarios} 
Following established methodologies~\cite{wang2022towards, zhang2021mining, zhou2023bootstrap}, we randomly partition each dataset into training, validation, and test sets at a ratio of 8:1:1. To assess algorithm performance in top-$k$ recommendation scenarios, we utilize standard evaluation metrics commonly employed in recommendation systems, namely Recall ($\mathrm{R}@k$) and Normalized Discounted Cumulative Gain (NDCG, shorted as $\mathrm{N}@k$). The parameter $k$ is set to 10 and 20. 
We employ two distinct data splitting strategies to evaluate our model under both general and cold-start conditions, following the protocols established by~\cite{zhang2021mining, zhou2023tale}.

\begin{table*}[bpt]
    \centering
    \def\arraystretch{0.9}	
    \setlength{\tabcolsep}{7.5pt}
    \caption{Performance comparison of different recommendation methods in terms of Recall@10 and NDCG@10. The best results are indicated in \textbf{bold} text, and the second-best results are \underline{underlined}. }
    \begin{tabular}{l| cc | cc | cc | cc | cc}
        \toprule
        Dataset & \multicolumn{2}{c|}{Baby} & \multicolumn{2}{c|}{Sports} & \multicolumn{2}{c|}{Clothing} & \multicolumn{2}{c|}{Electronics} & \multicolumn{2}{c}{Microlens}\\
        \midrule
        Metric & R@10 & N@10 & R@10 & N@10 & R@10 & N@10 & R@10 & N@10 & R@10 & N@10\\
        \midrule
        MF & 0.0357 & 0.0192 & 0.0432 & 0.0241 & 0.0206 & 0.0114 & 0.0235 & 0.0127 & 0.0624 & 0.0322 \\
        LightGCN & 0.0479 & 0.0257 & 0.0569 & 0.0311 & 0.0361 & 0.0197 & 0.0363 & 0.0204 & 0.0720 & 0.0376 \\
        SelfCF & 0.0521 & 0.0279 & 0.0630 & 0.0344 & 0.0415 & 0.0224 & 0.0442 & 0.0251  &  0.0723 &  0.0369 \\
        DirectAU & 0.0543 & 0.0300 & 0.0682 & 0.0379 & 0.0443 & 0.0240 & {0.0460} & \underline{0.0262}  & \underline{0.0817} & \underline{0.0429}\\
        \midrule
        VBPR & 0.0423 & 0.0223 & 0.0558 & 0.0307 & 0.0281 & 0.0158 & 0.0293 & 0.0159  & 0.0677 & 0.0351 \\
        MMGCN & 0.0421 & 0.0220 & 0.0401 & 0.0209 & 0.0227 & 0.0120 & 0.0207 & 0.0109  & 0.0421 & 0.0207 \\
        GRCN  & 0.0532 & 0.0282 & 0.0599 & 0.0330 & 0.0421 & 0.0224 & 0.0349 & 0.0194  & 0.0702 & 0.0365 \\
        LATTICE & 0.0547 & 0.0292 & 0.0620 & 0.0335 & 0.0492 & 0.0268 & OOM & OOM  & 0.0726 & 0.0380\\
        SLMRec & 0.0547 & 0.0285 & 0.0676 & 0.0374 & 0.0540 & 0.0285 & 0.0443 & 0.0249  & 0.0778 &  0.0405 \\
        BM3  & 0.0564 & 0.0301 & 0.0656 & 0.0355 & 0.0422 & 0.0231 & 0.0437 & 0.0247  & 0.0606 & 0.0304 \\
        FREEDOM  & 0.0627 & 0.0330 & {0.0717} & {0.0385} & {0.0629} & {0.0341} & 0.0396 & 0.0220  & 0.0674 & 0.0345 \\
        LGMRec & {0.0644} & {0.0349} & 0.0720 & 0.0390 & 0.0555 & 0.0302 & 0.0417 & 0.0233  & 0.0748 & 0.0390  \\
        DA-MRS & {0.0626} & {0.0339} & {0.0708} & {0.0379} & {0.0633} & {0.0342} & OOM & OOM  & {0.0801} & {0.0419} \\
        MIG-GT & \underline{0.0665} & \underline{0.0361} & \underline{0.0753} & \underline{0.0414} & \underline{0.0636} & \underline{0.0347} & \underline{0.0467} & {0.0261}  & {0.0806} & {0.0426} \\
        \midrule
        \textbf{\mmnm{}} & \textbf{0.0692}* & \textbf{0.0381}* & \textbf{0.0837}* & \textbf{0.0467}* & \textbf{0.0701}* & \textbf{0.0386}* & \textbf{0.0519}* & \textbf{0.0297}*  & \textbf{0.0852}* & \textbf{0.0450}* \\
        \textbf{\textit{Imp.}} & {4.06\%} & {5.54\%} & {11.16\%} & {12.80\%} & {10.22\%} & {11.24\%} & {11.13\%} & {13.36\%}  & {4.28\%} & {4.89\%} \\
        \bottomrule
        \multicolumn{11}{l}{- `OOM' denotes an Out-Of-Memory condition encountered on a Tesla V100 GPU with 32 GB of memory. }
    \end{tabular}
    \label{tab:perform10}
\end{table*}

\textit{\textbf{Warm-Start Evaluation}.}
For each user in the dataset, we implement a stratified random sampling approach to partition their historical interactions. The dataset is segregated into three mutually exclusive subsets: training, validation, and testing, with a ratio of 8:1:1, respectively. This methodology ensures: i). A minimum of five interactions per user in the processed dataset. ii). At least one sample for both validation and testing phases. iii). A minimum of three interactions for model training.

\textit{\textbf{Cold-Start Evaluation}.}
To simulate cold-start conditions, we adopt the following procedure: i). Random selection of 20\% of items from the complete item pool. ii). Equal bifurcation of the selected items into validation (10\%) and test (10\%) sets. iii). Assignment of user-item interactions to training, validation, or testing sets based on the item's designated partition. This approach ensures that items in the validation and test sets remain unseen during the training phase, accurately replicating the challenges inherent in cold-start scenarios where no prior information is available for a subset of items during the recommendation process.

\subsection{Implementation details}
We set the embedding dimension of $d$ as $d=64$ and utilize the Xavier initialization method~\cite{glorot2010understanding} for user embedding initialization. To minimize the proposed loss function, we optimize the model using the Adam optimizer~\cite{kingma2014adam} with a learning rate of $0.001$. For the baseline methods, we strictly adhere to the hyperparameter tuning procedures outlined in their respective original papers. Regarding our proposed method, we employ a grid search to identify the optimal combination of hyperparameters across all datasets. Specifically, we explore the trade-off $\gamma$ between alignment and uniformity loss within the range $[0.2, 3.0]$ with increments of $0.2$.
The model selection is based on the highest $\mathrm{R}@20$ score achieved on the validation data. The training process is limited to a maximum of 100 epochs, with early stopping implemented after 10 epochs. Our implementation is based on the MMRec framework~\cite{zhou2023mmrec}.

\section{Experiment Results}
\subsection{Performance Comparison}
\subsubsection{Warm-Start Evaluation of \mmnm{}}
Experimental results from different algorithms are presented in Table \ref{tab:perform} and Table~\ref{tab:perform10}, from which we observe the following phenomena. \textit{Firstly}, our proposed \mmnm{} achieves the best results in terms of Recall and NDCG across all datasets. 
Quantitatively, \mmnm{} achieved an average NDCG@20 improvement of 11.95\% over DA-MRS and 8.56\% over MIG-GT, respectively, when evaluated across all available datasets.
The consistent improvement over all baselines demonstrates the superiority of our \mmnm{}, even on the largest ``Electronics'' dataset. \textit{Secondly},the efficacy of incorporating multimodal information in recommendation models may be attenuated when applied to large-scale datasets.
For example, on ``Electronics'' dataset, almost all evaluated multimodal approaches except MIG-GT demonstrate inferior performance compared to DirectAU and SelfCF, highlighting notable limitations in their methodologies. This observation suggests that in larger datasets such as ``Electronics'', user-item interaction data assumes a more pivotal role in recommendation accuracy than in smaller datasets. 
The imposition of a uniformity loss between user-user and item-item pairs plays a crucial role in their differentiation. By leveraging this principle in conjunction with multimodal features, our proposed \mmnm{} framework demonstrates superior performance across all baselines, on the ``Electronics'' dataset. Specifically, \mmnm{} demonstrates a substantial improvement of 13.97\% in NDCG@20 compared to DirectAU on this dataset. This significant performance gain underscores the efficacy of our approach in integrating uniformity constraints with multimodal information for enhanced recommendation accuracy.

\subsubsection{Cold-Start Evaluation of \mmnm{}}
Multimodal recommendation models, by incorporating additional information beyond user-item interactions, mitigate the challenges posed by data sparsity. 
Fig.~\ref{fig:coldstart} illustrates the recommendation performance of our proposed \mmnm{} and three representative models.

\begin{figure}[bpt]
  \centering
  \includegraphics[width=\linewidth]{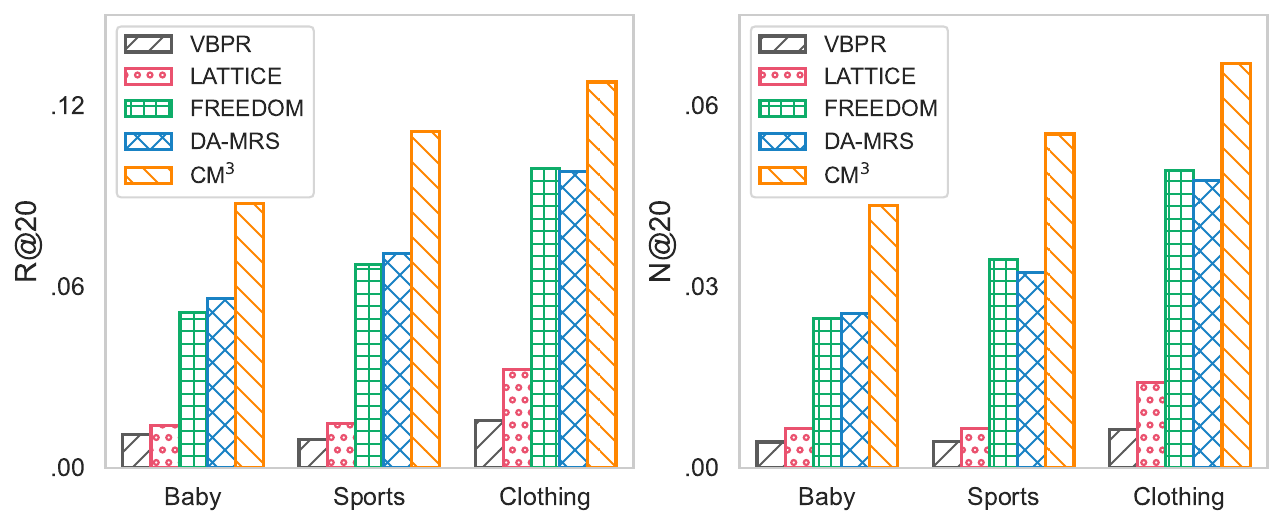}
  \vspace{-12pt}
  \caption[]{Performance of \mmnm{} compared with representative baselines under cold-start settings.}
  \label{fig:coldstart}
  \vspace{-14pt}
\end{figure}

The figure elucidates the subsequent facets:
\textit{Firstly}, incorporating multimodal features into the training loss function can enhance the robustness of recommendation models in cold-start scenarios. For example, VBPR concatenates multimodal features with item IDs for item representation learning. While LATTICE solely utilizes multimodal features to construct the item-item graph, VBPR achieves competitive performance on smaller datasets like ``Baby'' and ``Sports''. This suggests that directly incorporating multimodal features in the loss function might be beneficial.
\textit{Secondly}, GCNs have the potential to propagate information and gradients to unseen items (cold-start items) during training, even if they haven't been observed in user-item interactions. This can alleviate the cold-start problem, particularly when the user-item graph is large and sparsely connected. Partial validation for this can be observed on ``Clothing'' dataset in Fig.~\ref{fig:coldstart}.
\textit{Thirdly}, we observe that \mmnm{} significantly outperforms baseline models. In addition to the aforementioned advantages of \mmnm{}, we hypothesize that the calibrated uniformity loss plays a crucial role in adjusting the item distribution. This adjustment enables unseen items to receive a loss signal, thereby facilitating the learning of their representations.

\subsection{Ablation Study}
To gain a comprehensive understanding of \mmnm{}, we conduct ablation studies to investigate the impact of each component on recommendation performance. 
\begin{table}[bpt]
	\centering	
	\def\arraystretch{0.85}	
        \setlength{\tabcolsep}{3pt}
	\caption{Performance comparison of \mmnm{} variants under different component ablation settings.}
	\begin{tabular}{llcccc}
		\toprule
		{Dataset} & {Metric} & \mmnm{}\textsubscript{w/o F} & \mmnm{}\textsubscript{w LF} & \mmnm{}\textsubscript{w StdU} & \textbf{\mmnm{}} \\
		\midrule
		\multirow{2}{*}{Baby} & R@20 & 0.0956 & 0.0991 & 0.0890 & 0.1034\\
		& N@20 & 0.0446 & 0.0460 & 0.0414 & 0.0470 \\
		\midrule
		\multirow{2}{*}{Sports} & R@20 & 0.1164 & 0.1180 & 0.1120 & 0.1222 \\
		& N@20 & 0.0543 & 0.0539 &  0.0531 & 0.0567 \\
		\midrule
		\multirow{2}{*}{Clothing} & R@20 & 0.0981 & 0.0994 & 0.0993 & 0.1006 \\
		& N@20 & 0.0457 & 0.0458 & 0.0460 & 0.0463 \\
		\bottomrule
	\end{tabular}
	\vspace{-8pt}
	\label{tab:perform_com}	
\end{table}

\begin{table}[bpt]
	\centering
	\def\arraystretch{0.85}	
        \setlength{\tabcolsep}{3.2pt}
	\caption{Performance comparison of \mmnm{} variants under different unimodal/multimodal features.}
	\begin{tabular}{llcccc}
		\toprule
		{Dataset} & {Metric} & \mmnm{}\textsubscript{w/o V} & \mmnm{}\textsubscript{w/o T}& \textbf{\mmnm{}} & \textbf{\mmnm{}}\textsubscript{MLLM}  \\
		\midrule
		\multirow{2}{*}{Baby} & R@20 & 0.0847 & 0.0860 & 0.1034 & 0.1062\\
		& N@20 & 0.0378 & 0.0382 & 0.0470 & 0.0477 \\
		\midrule
		\multirow{2}{*}{Sports} & R@20 & 0.1035 & 0.1018 & 0.1222 & 0.1246 \\
		& N@20 & 0.0472 & 0.0465 & 0.0567 & 0.0576 \\
		\midrule
		\multirow{2}{*}{Clothing} & R@20 & 0.0725 & 0.0679 & 0.1006 & 0.1065\\
		& N@20 & 0.0336 & 0.0311 & 0.0463 & 0.0488\\
		\bottomrule
	\end{tabular}
	\vspace{-10pt}
	\label{tab:perform_mm}	
\end{table}

\subsubsection{Component Ablation} In this study, we explore the contributions of the spherical Bézier fusion and calibrated uniformity loss in comparison to the linear interpolated fusion and standard uniformity loss. We consider the following variants while fixing all other settings. 
\begin{itemize}[leftmargin=*]
    \item \textbf{\mmnm{}\textsubscript{w/o F}} indicates that we only remove the proposed fusion strategies during \mmnm{} training.  
    \item \textbf{\mmnm{}\textsubscript{w LF}} means that the multimodal features is fused with the conventional linear interpolation. 
    \item \textbf{\mmnm{}\textsubscript{w StdU}} represents that the calibrated uniformity loss is substituted by the standard uniformity loss. 
\end{itemize}
Table \ref{tab:perform_com} presents the experimental results of the aforementioned model variants across three datasets. Analysis of these results yields several noteworthy observations: i). The full version of our model consistently outperforms all ablation settings across every dataset and evaluation metric. This finding suggests that each component of the model contributes positively to its overall performance. 
ii). Each dataset reacts differently to the removal of model components. For example, on ``Baby'' and ``Sports'' datasets, removing the calibrated uniformity leads to a significant performance drop, whereas the drop on ``Clothing'' dataset is less pronounced. Linear fusion performs comparably to our method on  ``Clothing'' dataset, but shows inferior results on the other two datasets. iii). Comparative analysis between the full model and the variant without fusion reveals that the proposed spherical Bézier fusion serves as an effective default strategy for enhancing recommendation accuracy.

\begin{figure*}[bpt]
  \centering
  \includegraphics[width=0.95\linewidth]{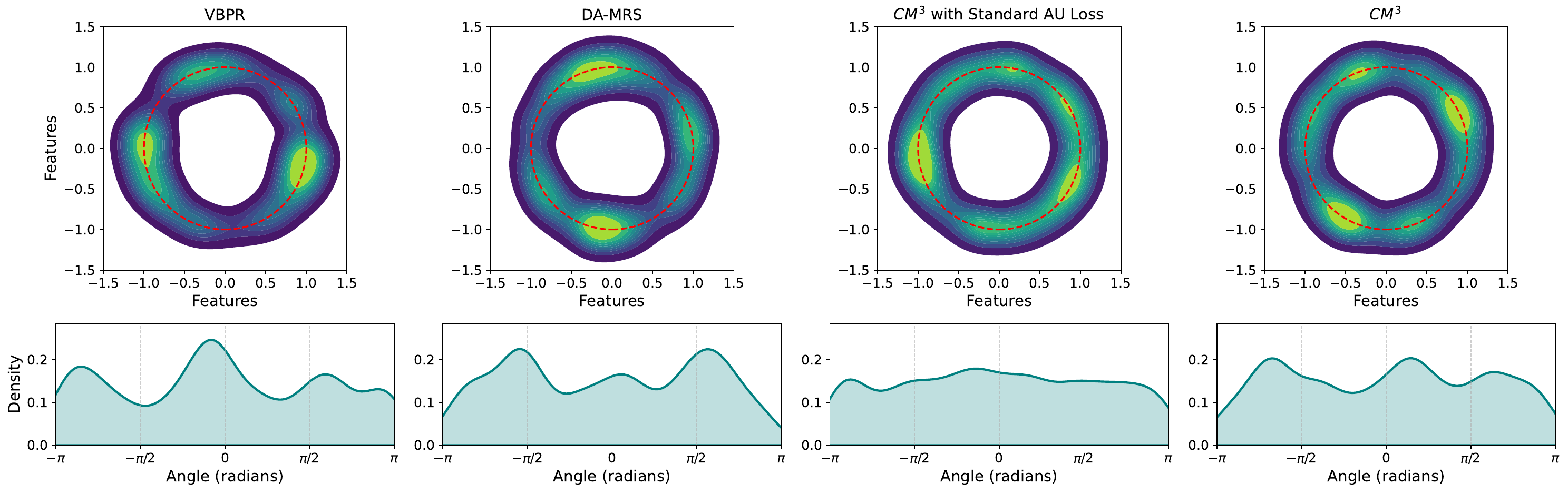}
  \caption[]{Distribution of item representations via KDE plot, with lighter areas indicating a higher concentration of points.}
  \label{fig:emb_kde}
  \vspace{-10pt}
\end{figure*}

\subsubsection{Multimodal Feature Ablation} In this study, we investigate the impact of unimodal features on recommendation performance. Specifically, we consider the following variants of \mmnm{}, either incorporating only unimodal features or utilizing features extracted with Multimodal Large Language Models (MLLMs).

Specifically, we leverage Meta's ``Llama-3.2-11B-Vision''~\cite{dubey2024llama} for converting visual content into textual captions. 
Text embeddings for items are subsequently generated from item captions using the ``e5-mistral-7b-instruct'' model~\cite{wang2024improving}. The resulting embedding vectors, each of dimension 4,096, are used to represent both image-derived and text-based item descriptions.

\begin{itemize}[leftmargin=*]
    \item \textbf{\mmnm{}\textsubscript{w/o V}} represents that \mmnm{} is trained without the visual features of items. 
    \item \textbf{\mmnm{}\textsubscript{w/o T}} denotes that \mmnm{} is trained without textual features. 
    \item \textbf{\mmnm{}\textsubscript{MLLM}} indicates that \mmnm{} is trained utilizing multimodal features derived from MLLMs.
\end{itemize}
Table \ref{tab:perform_mm} reports the recommendation accuracy of \mmnm{} and its two variants on three datasets. From experiment results, we observe that: i). Generally, \textbf{\mmnm{}\textsubscript{w/o V}}, which excluding the visual features clearly perform better than its counterpart \textbf{\mmnm{}\textsubscript{w/o T}} on ``Sports'' and ``Clothing'' datasets. This observation suggests that textual features are essential to ensure recommendation performance. ii). An exception is that \textbf{\mmnm{}\textsubscript{w/o T}} performs slightly better than \textbf{\mmnm{}\textsubscript{w/o V}} on  ``Baby'' dataset. We guess that product images in the Baby category may provide discriminative information for well modeling item representations. In summary, visual and textual features complement each other from different perspectives, allowing \mmnm{} to achieve the best results across all three datasets. 
It was further noted that \mmnm{}'s performance on the Clothing dataset was enhanced by 5.40\% in NDCG@20 via the use of features extracted from MLLMs.

\subsection{Item Representation Distribution}
To investigate how \mmnm{} enforces the distribution of item representations, we generated two plots in Fig.~\ref{fig:emb_kde} based on Sports dataset. The first shows feature distributions using Gaussian kernel density estimation (KDE) in $\mathbb{R}^2$, with lighter colors indicating a higher density of points. The second is a KDE plot of the angles, calculated as $\operatorname{arctan2}(y, x)$ for each point $(x, y) \in S^1$.
As depicted in the figure, the application of contrastive loss encourages a more uniform distribution among item representations (DA-MRS and \mmnm{} over VBPR). Notably, the proposed calibrated uniformity loss provides a fine-grained adjustment, mitigating the tendency towards excessive uniformity that can occur with standard uniformity loss.

\subsection{Hyperparameter Sensitivity Study}
\begin{figure}[bpt]
  \centering
  \includegraphics[width=0.95\linewidth, trim={10 10 10 10},clip]{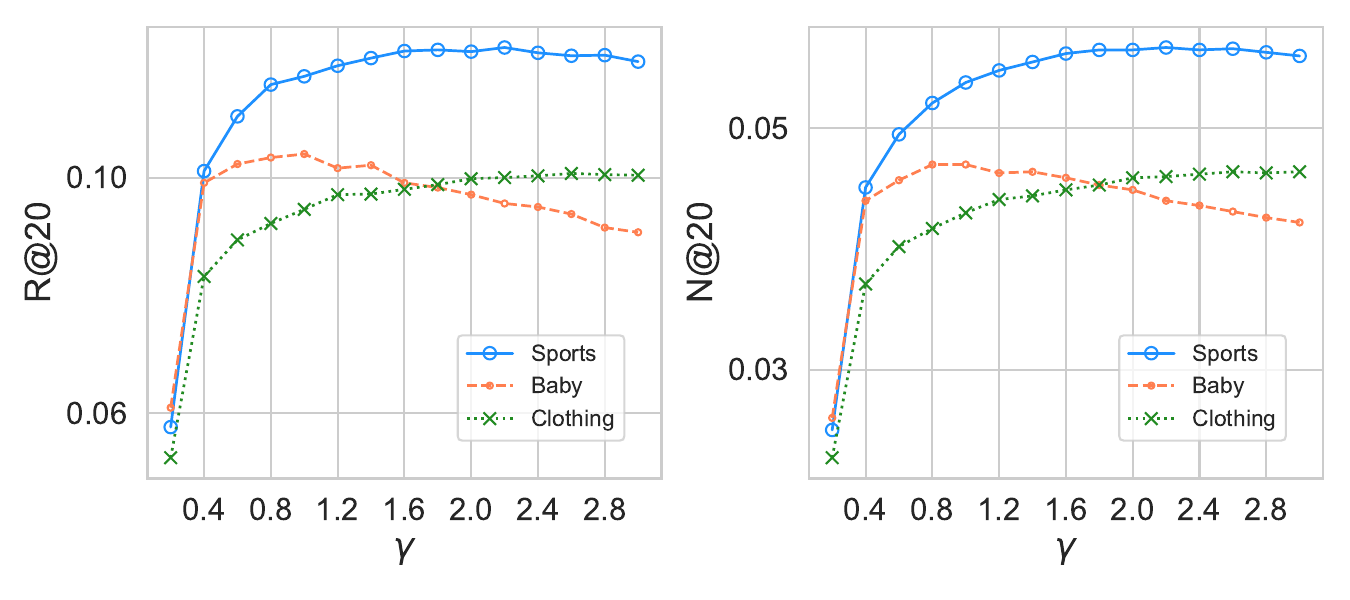}
  \caption[]{Performance analysis of \mmnm{} across varying alignment-uniformity trade-offs $\gamma$.}
  \label{fig:hyper}
  \vspace{-12pt}
\end{figure}

To investigate the influence of the trade-off factor $\gamma$ in the loss function, we conduct experiments to examine the sensitivity of \mmnm{} with respect to $\gamma$ across three datasets. From Fig.~\ref{fig:hyper}, we observe the following: i). As $\gamma$ increases from $0.4$ to $0.8$, the recommendation accuracy improves dramatically, suggesting that uniformity plays a crucial role in learning user and item representations. ii). 
The behavior of the datasets diverges notably when the parameter $\gamma$ exceeds 0.8.  
The metrics $R@20$ and $N@20$ maintain relatively stable and high values as $\gamma$ increases beyond 0.8. This phenomenon underscores the critical role of the uniformity loss in these larger datasets, suggesting that a higher degree of uniformity constraint continues to benefit model performance.
In contrast, our \mmnm{} model exhibits a gradual performance decline when $\gamma$ surpasses 0.8. 
This observation suggests that the optimal balance between alignment and uniformity for the smaller dataset is achieved at a lower $\gamma$ value. This aligns with previous research by~\cite{wang2022towards}, which demonstrated that excessive optimization of uniformity can be detrimental to recommendation performance.
Refer the Appendix for more details.

\section{Conclusion}
This study empirically elucidates the contrasting optimization dynamics of alignment and uniformity in contemporary multimodal recommender systems. We introduced a calibrated uniformity loss that incorporates inherent multimodal similarities, effectively refining the representational space and promoting a better affinity for similar items. Empirical evaluations across five datasets confirmed our model's superiority over existing baselines. Our findings highlight that uniformity plays a more pivotal role than alignment on large-scale datasets.
Furthermore, we demonstrate that calibrating item uniformity using multimodal features presents a viable approach to modulating the nuanced relations between items. This strategy effectively alleviates the inherent dilemma in current alignment and uniformity optimization paradigms.

\bibliographystyle{unsrt}
\bibliography{cm3}

\appendix

\section{Proof of Proposition 1}
\label{append:a}

Spherical Bézier fusion \cite{oh2024geodesic} augments data by mixing the visual representation $\vec{a}$ and textual representation $\vec{b}$ with Eq. (3).
In the following, we offer the formal proof. 
\begin{proof}
To determine whether $m_{\lambda}(\vec{a}, \vec{b})$ is a unit vector for arbitrary $\lambda$, we need to verify that the norm squared of $ m_{\lambda}(\vec{a}, \vec{b})$ is equal to 1, 
\begin{align*}
\| m_{\lambda}(\vec{a}, \vec{b}) \|^2 &= \frac{\left\| \vec{a} \sin(\lambda \theta) + \vec{b} \sin((1 - \lambda) \theta) \right\|^2}{\sin^2(\theta)}. \nonumber
\end{align*}

The numerator (denoted by $\mathcal{N}$) of the above equation can be expanded as,
\begin{align*}
\mathcal{N} &= \sin^2(\lambda \theta) \| \vec{a} \|^2 + 2 \sin(\lambda \theta) \sin((1 - \lambda) \theta) (\vec{a} \cdot \vec{b})\\ &+ \sin^2((1 - \lambda) \theta) \| \vec{b} \|^2 \\
&=\sin^2(\lambda \theta) + 2 \sin(\lambda \theta) \sin((1 - \lambda) \theta) \cos(\theta) + \sin^2((1 - \lambda) \theta),
\end{align*}
since $ \| \vec{a} \| = \| \vec{b} \| = 1 $ and $ \vec{a} \cdot \vec{b} = \cos(\theta)$. Using trigonometric identities, we further simplify the numerator $\mathcal{N}$ and obtain, 
\begin{align*}
\mathcal{N} &= \left[ \frac{1 - \cos(2 \lambda \theta)}{2} + \frac{1 - \cos(2(1 - \lambda) \theta)}{2} \right] \\
& + \left[ \cos((2 \lambda - 1) \theta) - \cos(\theta) \right] \cos(\theta) \\
&= 1 - \cos(\theta) \cos((2 \lambda - 1) \theta) + \cos(\theta) \cos((2 \lambda - 1) \theta) - \cos^2(\theta) \\
&= \sin^2(\theta).
\end{align*}
Therefore, $\| m_{\lambda}(\vec{a}, \vec{b}) \|^2 = \frac{\mathcal{N}}{\sin^2(\theta)} = \frac{\sin^2(\theta)}{\sin^2(\theta)} = 1.$
We conclude that $m_{\lambda}(\vec{a}, \vec{b})$ is a unit vector when $ \vec{a} $ and $ \vec{b} $ are unit vectors.
Consequently, we assert that its extended version with Eq. (2) in the main paper, which integrates all available modality features, also lies on the hypersphere, provided that each individual feature lies on the hypersphere.
Thus, Proposition 1 holds.
\end{proof}

\section{Learning Curves of Alignment and Uniformity in \mmnm{}} 
Fig.~\ref{fig:curve} illustrates the learning curves of alignment and uniformity losses in \mmnm{}. During the initial epochs, recommendation performance increases sharply as both the alignment loss and calibrated uniformity loss of items decrease. However, model performance fluctuates and even slightly declines on the Clothing dataset when the uniformity loss of items levels off, despite the continuous decline in alignment loss. This observation suggests that minimizing alignment loss alone, without calibrated uniformity loss, does not necessarily lead to better performance.

\begin{figure}[bpt]
  \centering
  \includegraphics[width=\linewidth]{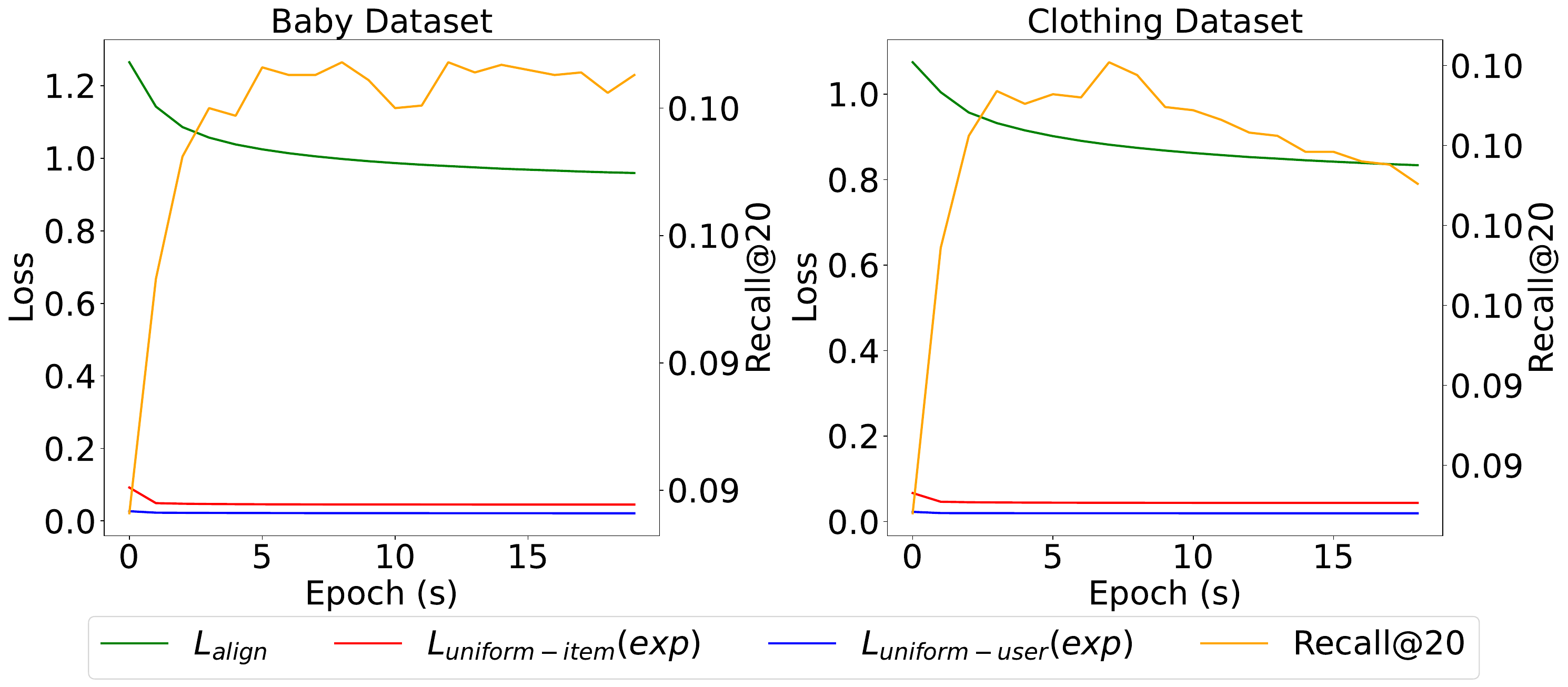}
  \caption[]{Learning curves of our \mmnm{}. Uniformity losses are close to $0$ due to the exponentiation of negative values.}
  \label{fig:curve}
  \vspace{-5pt}
\end{figure}

\section{Item Uniformity Assessment with Different Item Features} 
\label{append:random_uni}
To evaluate the impact of multimodal features on item uniformity, we utilized randomly generated features as a contrast. As demonstrated in Table~\ref{tab:mm_uni_random}, original features yielded inferior uniformity, necessitating greater optimization efforts, as visualized in Fig.1 of the main paper.
\begin{table}[h!]
	\centering	
        \setlength{\tabcolsep}{3pt}
	\caption{Item uniformity under various features.}
    \vspace{-6pt}
	\begin{tabular}{lcccccc}
		\toprule
		{Item Feature} & {VBPR} & BM3 & LGMRec & DA-MRS & FREEDOM \\
		\midrule
		{Multimodal} & -6.00 & -10.37 & -15.14 & -16.96 &  -17.42  \\
		\midrule
		{Random} & -8.19 & -10.46 & -17.08 &  -17.13 & -21.05 \\
		\bottomrule
	\end{tabular}
	\label{tab:mm_uni_random}	
\end{table}

\section{Concrete Runtime Comparison}

We further evaluate the concrete runtime performance of our proposed model against several representative models on the \textbf{LARGEST} available dataset (Electronics) from the multimodal recommendation literature. As shown in the following Table~\ref{tab:runtime_mm}, which reveals: (i) Multimodal models generally require greater memory and training time compared to non-multimodal counterparts; (ii) The memory usage and training time of our model are comparable to those of other multimodal models (\eg MMGCN, FREEDOM, GRCN), yet our model achieves superior recommendation performance.
Beyond its other benefits, the proposed model achieves quicker convergence than most multimodal models, leading to lower overall training expenses.
\begin{table}[h!]
	\centering	
	\caption{Runtime cost of recommender models.}
    \vspace{-6pt}
	\begin{tabular}{lccccc}
		\toprule
		& MMGCN & GRCN &  MIG-GT & \mmnm{}\\
		\midrule
		{Memory (G)} &14.54 & 17.38 & 16.85 &  11.32  \\
		\midrule
		{Time/Epoch ($sec.$)} & 470.15 & 152.68 &  34.19 & 202.12  \\
		{Convergent Epoch} & 22 & 151 &  351 & 26  \\
		{Train Time $\approx$ ($hour$)} & 2.87 & 6.40 &  3.33 & 1.46  \\
		\bottomrule
	\end{tabular}
	\label{tab:runtime_mm}	
\end{table}

\end{document}